\documentclass[%
preprint,
amsmath,
amssymb,
 aps,
prl,
]{revtex4-2}

\usepackage{setspace}

\usepackage[LGR,T1]{fontenc} 
\usepackage[utf8]{inputenc} %
\usepackage{substitutefont} 
\usepackage[polutonikogreek,english]{babel}
\usepackage[largesc]{newtxtext} %
\usepackage[varqu,varl]{zi4}
\usepackage{cabin}
\usepackage[vvarbb,upint,slantedGreek,smallerops]{newtxmath}
\substitutefont{LGR}{\rmdefault}{Tempora} 

\usepackage{textcomp}

\usepackage[margin=1.5in]{geometry}
\usepackage{hyperref}
\usepackage[hyphens]{xurl}
\usepackage{amsmath}
\usepackage{amsfonts}
\usepackage{amssymb}

\usepackage[pdftex]{graphicx}
\usepackage{lastpage}
\usepackage{slashed}

\usepackage{xcolor}

\usepackage{xfrac}

\usepackage{comment}
\usepackage[yyyymmdd,hhmmss]{datetime} 

\usepackage{lastpage}
\usepackage[running,mathlines]{lineno}
\usepackage{isomath}
\usepackage{siunitx}
\usepackage{upgreek}

\usepackage{marginnote}

\usepackage{amsthm}

\newtheorem{theorem}{Theorem}

\newtheorem{lemma}{Lemma}

\newtheorem*{law*}{Law}

\theoremstyle{definition}
\newtheorem{definition}{Definition}

\usepackage{verbatim}

\usepackage{xspace}
\usepackage{braket}
\usepackage{mathtools}
\mathtoolsset{showonlyrefs=true}

\usepackage{siunitx}
\usepackage{physics}

\newcommand{\piu}%
{\textrm{\greektext p}}
\newcommand{\eu}%
{\ensuremath{\mathrm{e}}}
\newcommand{\iu}%
{\ensuremath{\mathrm{i}}}

\providecommand{\newoperator}[3]{%
\newcommand*{#1}{\mathop{#2}#3}}

\newcommand{\tran}%
{\mathsf{T}}
\newcommand{\herm}%
{\textsf{H}}
\newcommand{\deltau}%
{\textrm{\greektext d}}
\newcommand{\Deltau}%
{\textrm{\greektext D}}

\makeatletter
\providecommand*{\diff}%
{\@ifnextchar^{\DIfF}{\DIfF^{}}}
\def\DIfF^#1{%
\mathop{\mathrm{\mathstrut d}}
\nolimits^{#1}\gobblespace}
\def\gobblespace{%
\futurelet\diffarg\opspace}
\def\opspace{%
\let\DiffSpace\!%
\ifx\diffarg(%
\let\DiffSpace\relax
\else
\ifx\diffarg[%
\let\DiffSpace\relax
\else
\ifx\diffarg\{%
\let\DiffSpace\relax
\fi\fi\fi\DiffSpace}

\newoperator{\loglike}%
{\mathrm{loglike}}{\nolimits}

\newoperator{\notloglike}%
{\mathrm{notloglike}}{\limits}

\makeatother

\newcommand{\normalx}[3]{\ensuremath{\mathscr{N}\left(#1 \mid #2,#3\right)}\xspace}

\renewcommand{\piu}{\uppi}
\renewcommand{\deltau}{\updelta}

\let\originalpartial\partial
\let\partial\relax
\newrobustcmd*{\partial}{\text{\rotatebox[origin=t]{10}{\scalebox{0.95}[1]{\ensuremath{\originalpartial}}}}\hspace{-0.05em}}

\usepackage{graphicx}
\usepackage{dcolumn}

\newcommand{\ee}{\ensuremath{\text{e}}\xspace}

\newcommand{\planckbar}{\ensuremath{\hbar}\xspace}

\newcommand{\setConstant}{\ensuremath{\mathscr{C}}\xspace}
\newcommand{\setDecreasing}{\ensuremath{\mathscr{D}}\xspace}
\newcommand{\setIncreasing}{\ensuremath{\mathscr{I}}\xspace}
\newcommand{\setOscillating}{\ensuremath{\mathscr{O}}\xspace}

\newcommand{\hessian}{\ensuremath{\mathbf{H}}\xspace}

\usepackage[normalem]{ulem}

\usepackage{bm}

\newcommand{\Identitymatrix}{\ensuremath{\mathbf{I}}\xspace}

\newcommand{\Sigmabold}{\ensuremath{\bm{\mathrm{\Sigma}}}\xspace}

\newcommand{\kernBeforeIntegral}{\ensuremath{\kern-0.17em}\xspace}

\newcommand{\mechanics}{mechanics\xspace}

\newcommand{\eqdefA}{=\xspace}

\newcommand{\entropyS}{\ensuremath{\mathrm{S}\xspace}}

\newcommand{\QCurve}{QCurve\xspace}
\newcommand{\QCurves}{QCurves\xspace}

\newcommand{\setAll}{\ensuremath{\mathscr{E}}\xspace}
\newcommand{\CPT}{\ensuremath{\mathrm{CPT}}\xspace}

\usepackage[mathscr]{euscript}

\usepackage[en-US]{datetime2}

\usepackage[inline]{enumitem}

\newcommand{\state}{\ensuremath{\text{state}}\xspace}
\usepackage{changepage}
\usepackage{nccmath}

\definecolor{myblue}{rgb}{0.25, 0.6, 0.8}
\definecolor{myred}{rgb}{0.9, 0.5, 0.0}

\begin{document}

\thispagestyle{empty}
\newpage

\title{
Quantum Entropy Evolution}
\author{{Davi Geiger} and {Zvi M.\ Kedem}}
\affiliation{ Courant Institute of Mathematical Sciences\\
  New York University, New York, New York 10012}
\setcounter{page}{1}

\begin{abstract}
A quantum coordinate-entropy formulated in quantum phase space has been recently proposed together with an entropy law that asserts that such entropy can not decrease over time. The coordinate-entropy is dimensionless,  a relativistic scalar, and it is invariant under  coordinate and CPT transformations. We study here the time evolution of this  entropy. 

We show that the entropy associated with coherent states evolving under a Dirac Hamiltonian is increasing. However, for the collisions of two particles, where each is evolving as a coherent state,  as they come closer to each other their spatial entanglement causes the total system's entropy to oscillate. We augment time reversal with time translation and show that CPT with time translation can transform particles with decreasing entropy  for a finite time interval  into anti-particles with increasing entropy  for the same finite time interval.  

We then analyze the impact of the entropy law for the evolution scenarios described above and explore the possibility that  entropy oscillations  trigger  the annihilation and the creation of particles.

\end{abstract}

\maketitle

\pagebreak

\tableofcontents

\pagebreak

\section{Introduction}

A time arrow emerges in physics only when a probabilistic behavior of ensembles of particles is considered in classical physics. In contrast,  quantum physics is presented as time reversible even though a probabilistic behavior is intrinsic even to a single-particle.   In ~\cite{GeigerKedem2021b} we proposed a definition of quantum entropy to measure the randomness of a quantum state, while accounting for all its degrees of freedom (DOFs). That entropy is the sum of two components: the coordinate-entropy and the spin-entropy, each defined in its own quantum phase space. 
The quantum phase space is the space of the projection of a state onto a pair  of conjugate basis. The spin entropy was elaborated in  ~\cite{geiger2021spin}.  We analyzed there the possible entropy evolution and conjectured that a law analogous to the classical second law of thermodynamics holds, applicable to all  particle physics.

Here we study in  more technical depth   the time evolution of the coordinate-entropy further address the issues studied in~\cite{GeigerKedem2021b}.  Following the results described above, we review a conjectured entropy law that the entropy of a quantum system is an increasing function of time. The results are applicable to both the Quantum Mechanics (QM) and the Quantum Field Theory (QFT) settings, but we  generally present them in only a more convenient setting.

\subsection{Our previous work on quantum coordinate-entropy}
\label{sec:quantum-entropy-def}
Given a state $\ket{\psi}_t$ and its density operator $\rho_t=\ket{\psi}_t\bra{\psi}_t$, we consider the quantum coordinate phase space to be the space of simultaenous projections of all possible state densities onto the basis $\ket{\mathbf{r}},\ket{\mathbf{p}}$, i.e.,  the state   $\ket{\psi}_t$  is described in quantum phase space by the pair $(\bra{\mathbf{r}}\rho_t\ket{\mathbf{r}},\bra{\mathbf{p}}\rho_t\ket{\mathbf{p}})$. The coordinate-entropy in quantum phase space was defined in~\cite{GeigerKedem2021b}  as
\begin{align}
  \entropyS
 &=-\kernBeforeIntegral\int  \, \rho_{\mathrm{r}} (\mathbf{r},t)  \rho_k ( \mathbf{k},t)\,
 \, \ln \left ( \rho_{\mathrm{r}} (\mathbf{r},t)  \rho_k (\mathbf{k},t) \,  \right) \,
 \diff^3\mathbf{r}\, \diff^3\mathbf{k}\,,
 \label{eq:Relative-entropy-scalar}
\end{align}
where   $\entropyS_{\mathrm{r}}=  -\kernBeforeIntegral\int  \rho_{\mathrm{r}} (\mathbf{r},t) \ln \rho_{\mathrm{r}} (\mathbf{r},t) \, \diff^3\mathbf{r}$, and analogously for $\entropyS_{\mathrm{k}}$, $\rho_{\mathrm{r}}(\mathbf{r},t)=\bra{\mathbf{r}}\rho_t\ket{\mathbf{r}} =|\psi(\mathbf{r},t)|^2$ and $\rho_k(\mathbf{k},t)=\bra{\mathbf{k}}\rho_t\ket{\mathbf{k}}=|\tilde \phi(\mathbf{k},t)|^2$,  with $\psi(\mathbf{r},t)$ and
$\tilde \phi(\mathbf{k},t)$ representing in QM the wave function and in QFT the coefficients of the Fock states. The momentum  is described by the change of variables  $\mathbf{p}= \hbar \mathbf{k}$, so that the entropy is dimensionless and invariant under changes of the units of measurements.

A natural extension of this  entropy to an $N$-particle QM system is
 \begin{align}
     S &= - \kernBeforeIntegral\int \diff^3 \mathbf {r}_1 \diff^3\mathbf {k}_1 \hdots \diff^3 \mathbf {r}_N \diff^3\mathbf {k}_N \,   \rho_{\mathrm{r}}(\mathbf {r}_1,\hdots, \mathbf {r}_N ,t)  \rho_{\mathrm{k}}(\mathbf {k}_1,\hdots, \mathbf {k}_N ,t)
     \\
     & \qquad \quad \times     \ln \left (  \rho_{\mathrm{r}}(\mathbf {r}_1,\hdots, \mathbf {r}_N ,t) \rho_{\mathrm{k}}(\mathbf {k}_1,\hdots, \mathbf {k}_N ,t) \right)
     \\
     &= -\kernBeforeIntegral\int \diff^3 \mathbf {r}_1 \hdots  \kernBeforeIntegral\int \diff^3 \mathbf {r}_N \,   \rho_{\mathrm{r}}(\mathbf {r}_1,\hdots, \mathbf {r}_N ,t)
     \ln \rho_{\mathrm{r}}(\mathbf {r}_1,\hdots, \mathbf {r}_N ,t)
     \\
     & \quad -  \kernBeforeIntegral\int \diff^3 \mathbf {k}_1 \hdots  \kernBeforeIntegral\int \diff^3 \mathbf {k}_N\,   \rho_{\mathrm{k}}(\mathbf {p}_1,\hdots, \mathbf {k}_N ,t)
     \ln \rho_{\mathrm{k}}(\mathbf {k}_1,\hdots, \mathbf {k}_N ,t) \, ,
     \label{eq:entropy-many-particles}
 \end{align}
 where $\rho_{\mathrm{r}}(\mathbf {r}_1,\hdots, \mathbf {r}_N ,t)=|\psi(\mathbf {r}_1,\hdots, \mathbf {r}_N ,t)|^2$ and $\rho_{\mathrm{k}}(\mathbf {k}_1,\hdots, \mathbf {k}_N ,t)=|\phi(\mathbf {k}_1,\hdots, \mathbf {k}_N ,t)|^2 $ are defined in QM via the projection of the state $\ket{\psi_t}^N$ of $N$ particles (the product of $N$ Hilbert spaces) onto the position $\bra{\mathbf {r}_1}\hdots\bra{\mathbf {r}_N} $ and the momentum $\bra{\mathbf {k}_1}\hdots\bra{\mathbf {k}_N} $ coordinate systems.

\section{Time Evolution: \QCurves and Entropy-Partition}
\label{sec:QCurve}

We introduce the concept of  a \QCurve to specify
a curve (or path) in a Hilbert space parametrized by time.  In QM a \QCurve is represented by  a triple $\big(\ket{\psi_{0}}, U(t), \deltau t \big)$ where $\ket{\psi_{0}}$ is the initial state, $U(t)=\eu^{-\iu \frac{H}{\hbar} t}$ is the evolution operator, and $[0, \deltau t]$ is the time interval of the evolution. Of course, one may also represent the initial state by a triple $\big(\rho_0, U(t), \deltau t \big)$, where $\rho_0=\ket{\psi_{0}}\bra{\psi_{0}}$ is the density matrix. Alternatively, we can represent the initial state in the quantum coordinate phase space by   $(\bra{\mathbf{r}}\rho_0\ket{\mathbf{r}}, \bra{\mathbf{k}}\rho_0\ket{\mathbf{k}})$ or $(\bra{\mathbf{r}}\ket{\psi_{0}}, \bra{\mathbf{k}}\ket{\psi_{0}})$, and in QFT by  $\left ( \bra{\state}\rho_r(\mathbf{r},0)\ket{\state}, \bra{\state}\rho_k(\mathbf{k},0)\ket{\state}\right)$, with $\rho_r(\mathbf{r},0)=\Psi^{\dagger}(\mathbf{r},0)\Psi(\mathbf{r},0)$ and $\rho_k(\mathbf{k},0)=\Phi^{\dagger}(\mathbf{k},0)\Phi(\mathbf{k},0)$. We will use any of these representations to describe a \QCurve as more convenient for manipulations for the problem at hand.

\begin{definition}[Partition of $\mathscr{E}$]
  \label{def:partition}
Let $\setAll$ to be the set of all the \QCurves.  We define a partition of  $\setAll$ based on the entropy evolution into four blocks:
 \begin{adjustwidth}{-1em}{}
\begin{description}
 \begin{description}
 \item[$\setConstant$]
Set of the \QCurves for which the entropy is a constant.
  \item[$\setIncreasing$]
Set of  the \QCurves for which the entropy is increasing, but it is not a constant.
 \item[$\setDecreasing$]
Set of the \QCurves for which the entropy is decreasing, but it is not a constant.
\item[$\setOscillating$]
Set of the oscillating \QCurves, with the entropy strictly increasing in some subinterval of $[0, \deltau t]$ and strictly decreasing in another subinterval of $[0,\deltau t]$.
 \end{description}
\end{description}
\end{adjustwidth}
 \end{definition}

Consider stationary states $\ket{\psi_t}=\ket{\psi_E} \ee^{-\iu \omega t} $ with $\omega={E}/{\planckbar}$, where  $E$ is an energy eigenvalue of the Hamiltonian, and $\ket{\psi_E}$ is the time-independent eigenstate of the Hamiltonian  associated with  $E$.

\begin{theorem}
\label{prop:stationaryStatesAreInC}
{All  stationary states  are in  $\setConstant$}.
\end{theorem}
\begin{proof}
Follows from the time invariance of the probabilities $\rho_t=\ket{\psi_t}\bra{\psi_t}=\ket{\psi_E}\bra{\psi_E}$.
\end{proof}

\section{The Coordinate-Entropy of Coherent States Increases With Time}
\label{sec:coherent-states}

Dirac's  free-particle Hamiltonian in QM \cite{dirac1930theory}   is
\begin{align}
  H = -\iu \planckbar\gamma^0 \vec{\gamma} \cdot \nabla+ m c \gamma^0 \, .
  \label{eq:Dirac-Hamiltonian}
\end{align}
It can be diagonalized in the
spatial Fourier domain $\ket{\mathbf {k}}$ basis  to obtain
\begin{align}
\omega(\mathbf{k})& =\pm c \sqrt{ \matrixsym{k}^2+\frac{m^2}{\planckbar^2} c^2}\,,
\label{eq:Fourier-Hamiltonians}
\end{align}
where $\omega(\mathbf{k})$ is the frequency component of the Hamiltonian.  We focus on the positive energy solutions and so the group velocity becomes
 \begin{align}
\mathbf{v_g}(\mathbf{k})&=\nabla_{\mathrm{k}} \omega(\mathbf{k})=\frac{\planckbar}{m}\frac{\mathbf {k}}{\sqrt{1+(\frac{\planckbar \matrixsym{k}}{m c})^2}}\,.
\label{eq:Fourier-group-velocity}
 \end{align}
In  \eqref{eq:w-dispersion} we will use the Taylor expansion of~\eqref{eq:Fourier-Hamiltonians} up to the second order, thus  requiring  the Hessian $\hessian(\mathbf{k})$, with the entries
\begin{align}
\hessian_{ij}(\mathbf{k})&= \frac{\partial^2 \omega(\mathbf{k}) }{\partial \matrixsym{k}_i \partial \matrixsym{k}_j} = \frac{\planckbar}{m} \left (1+\left (\frac{\planckbar \matrixsym{k}}{m c}\right)^2\right)^{-\frac{3}{2}}\left [ \deltau_{i,j}\left (1+\left (\frac{\planckbar \matrixsym{k}}{m c}\right)^2 \right) -\left (\frac{\planckbar \matrixsym{k}_i}{m c}\right) \left (\frac{\planckbar \matrixsym{k}_j}{m c}\right)\right ]\quad
\label{eq:Fourier-group-Hessian}
\end{align}
for the positive energy solution.  The three  (positive) eigenvalues of $\hessian(\mathbf{k})$  are
\begin{align}
\label{eigenvalues-Hamiltonian}
\lambda_1&=\frac{\planckbar}{m}\left (1+\left (\frac{\planckbar \matrixsym{k}}{m c}\right)^2\right)^{-\frac{3}{2}}=\, \planckbar \frac{m^2}{\left (m^{2}+\mu^2(\matrixsym{k})\right)^{\frac{3}{2}}}\,,
\\
\lambda_{2,3}& = \frac{\planckbar}{m} \left (1+\left (\frac{\planckbar \matrixsym{k}}{m c}\right)^2\right)^{-\frac{1}{2}}=\,
 \planckbar \frac{1}{(m^2+\mu^2(\matrixsym{k}))^{\frac{1}{2}} }\, ,
\end{align}
where $\mu(\matrixsym{k})={\planckbar \matrixsym{k}}/{c}$ is  the kinetic energy in mass units. The Hessian is positive definite for positive energy, and so it gives  a measure of the dispersion of the wave.

We now consider initial solutions that are localized in space, $\psi_{\mathrm{k}_0}(\mathbf {r}-\mathbf {r}_0)=\psi_0(\mathbf {r}-\mathbf {r}_0) \, \eu^{\iu \mathbf {k}_0\cdot \mathbf {r}}$, where $\mathbf{r}_0$ is the mean value of $\mathbf {r}$. Assume that the variance, $\int \diff^3\mathbf {r}\, (\mathbf {r}-\mathbf {r}_0)^2 \rho_{\mathrm{r}}(\mathbf {r}) $, is finite, where  $\rho_{\mathrm{r}}(\mathbf {r})=|\psi_0(\mathbf {r})|^2$.  In a Cartesian representation, we can write the initial state in the spatial frequency domain as $ \phi_{\mathrm{r}_0}(\mathbf {k}-\mathbf {k}_0)= \phi_{0}(\mathbf {k}-\mathbf {k}_0) \, \eu^{-\iu (\mathbf {k}-\mathbf {k}_0)\cdot \mathbf {r}_0}$, where $\phi_{0}(\mathbf {k})$ is the Fourier transform of $\psi_0(\mathbf {r})$, and so
the variance of
$\rho_{\mathrm{k}}(\mathbf {k})=|\phi_{\mathrm{r}_0}(\mathbf {k}-\mathbf {k}_0)|^2$
is also finite, with the mean in the  spatial frequency center $\mathbf {k}_0$.

The time evolution of $\psi_{\mathrm{k}_0}(\mathbf {r}-\mathbf {r}_0)$  according a Hamiltonian  with a dispersion relation $\omega(\mathbf {k})$, and written  via the inverse Fourier transform, is
\begin{align}
 \psi_{\mathrm{k}_0}(\mathbf {r}-\mathbf {r}_0,t) &={\frac {1}{({\sqrt {2\piu }})^{3}}}\int \Phi_{\mathrm{r}_0} (\mathbf {k}-\mathbf {k}_0) \, \eu^{-\iu \omega(\mathbf {k}) t} \eu^{\iu \mathbf {k} \cdot \mathbf {r} }\diff^{3}\mathbf {k}\,.
 \label{eq:time-evolution-psi-general}
\end{align}
As $\phi_{\mathrm{r}_0}(\mathbf {k}-\mathbf {k}_0)$ fades away exponentially from $\mathbf {k}=\mathbf {k}_0$, we expand~\eqref{eq:Fourier-Hamiltonians} in a Taylor series and approximate it as
\begin{align}
 \omega(\mathbf {k})   &\approx
 \mathbf{\mathbf{ v_p}}(\mathbf {k}_0) \cdot \mathbf {k}_0+\mathbf{v_g}(\mathbf {k}_0) \cdot (\mathbf {k}-\mathbf{k}_0)
+
 \frac{1}{2}(\mathbf {k}-\mathbf {k}_0)^{\tran} \hessian(\mathbf {k}_0)\,
 (\mathbf {k}-\mathbf {k}_0)\, , \,
  \label{eq:w-dispersion}
\end{align}
where $\mathbf{\mathbf{ v_p}}(\mathbf {k}_0)$, $\mathbf{ \mathbf{v_g}}(\mathbf {k}_0)$, and $ \hessian(\mathbf {k}_0) $ are the phase velocity $\frac{\omega(\mathbf {k}_0)}{|\mathbf {k}_0|} \hat{\mathbf {k}}_0$, the group velocity \eqref{eq:Fourier-group-velocity}, and the Hessian \eqref{eq:Fourier-group-Hessian} of the dispersion relation $\omega(\mathbf {k})$, respectively. Then after inserting~\eqref{eq:w-dispersion}    into~\eqref{eq:time-evolution-psi-general}, we obtain the quantum dispersion transform
\begin{flalign}
 \phi_{\mathbf {r}_{\mathbf {k}_0}^t} (\mathbf {k}-\mathbf {k}_0, t) & \approx   \frac{1}{Z_k} \, \eu^{-\iu t \mathbf{\mathbf{ v_p}}(\mathbf {k}_0) \cdot \mathbf {k}_0} \, \Phi_{\mathbf {r}_{\mathbf {k}_0}^t} (\mathbf {k}-\mathbf {k}_0)  \,  \normalx{\mathbf {k}}{\mathbf{k}_0}{-\iu t^{-1} \hessian^{-1}(\mathbf {k}_0) } ,
 \\
\psi_{\mathrm{k}_0}(\mathbf {r}-\mathbf {r}_{\mathbf {k}_0}^t,t)  & \approx
 \frac{1}{Z_{\mathrm{r}}}\eu^{-\iu t \mathbf{\mathbf{ v_p}}(\mathbf {k}_0) \cdot \mathbf {k}_0} \,  \psi_{\mathrm{k}_0}(\mathbf {r}- \mathbf {r}_{\mathbf {k}_0}^t) \ast \normalx{\mathbf {r}}{\mathbf {r}_{\mathbf {k}_0}^t}{\iu t \hessian(\mathbf {k}_0) }, \, \,
 \label{eq:time-evolution-psi-dispersion}
\end{flalign}
where $\mathbf {r}_{\mathbf {k}_0}^t=\mathbf {r}_0+\mathbf{v_g}(\mathbf {k}_0) t$, $\Phi_{\mathbf {r}_{\mathbf {k}_0}^t} (\mathbf {k}-\mathbf {k}_0) = \phi_{0}(\mathbf {k}-\mathbf {k}_0) \, \eu^{-\iu (\mathbf {k}-\mathbf {k}_0)\cdot \mathbf {r}_{\mathbf {k}_0}^t}$, with Fourier transform $\psi_{\mathrm{k}_0}(\mathbf {r}- \mathbf {r}_{\mathbf {k}_0}^t)  $; $\ast$ denotes a convolution, $Z_{r}$ and $Z_{k}$ normalize the amplitudes, and $\mathscr{N}$ is a normal distribution. Consequently, $\psi_{\mathrm{k}_0}(\mathbf {r}-\mathbf {r}_{\mathbf {k}_0}^t,t)$ is the spatial Fourier transform of $\Phi_{\mathbf {r}_{\mathbf {k}_0}^t} (\mathbf {k}-\mathbf {k}_0, t) $.

The probability densities associated with the probability amplitudes in \eqref{eq:time-evolution-psi-dispersion} are
\begin{align}
 \rho_{\mathrm{r}}(\mathbf {r}-\mathbf {r}_{\mathbf {k}_0}^t,t) & = \frac{1}{Z_{\mathrm{r}}^2}
 | \psi_{\mathrm{k}_0}(\mathbf {r}- \mathbf {r}_{\mathbf {k}_0}^t) \ast \normalx{\mathbf {r}}{\mathbf {r}_{\mathbf {k}_0}^t}{\iu \, t\, \hessian(\mathbf {k}_0) }|^2\,,
 \\
\rho_{\mathrm{k}}(\mathbf {k}-\mathbf {k}_0,t) & = \frac{1}{Z_k^2}|\Phi_{\mathbf {r}_{\mathbf {k}_0}^t} (\mathbf {k}-\mathbf {k}_0)|^2   \, .
 \label{eq:time-evolution-rho-dispersion}
\end{align}

\begin{lemma}[Dispersion Transform and Reference Frames]
\label{lemma:simplified-densities}
The entropy associated with ~\eqref{eq:time-evolution-rho-dispersion}
is equal to the entropy associated with the simplified probability densities
\begin{align}
 \rho^\entropyS_{\mathrm{r}}(\mathbf {r},t)&=  \frac{1}{Z^2}| \psi_{0}(\mathbf {r}) \ast \normalx{\mathbf {r}}{0}{\iu \, t\, \hessian(\mathbf {k}_0) }|^2\,,
 \\
\rho^\entropyS_{\mathrm{k}}(\mathbf {k},t) & = \frac{1}{Z_k^2}|\Phi_{0} (\mathbf {k})|^2 =\rho^\entropyS_{\mathrm{k}}(\mathbf {k},t=0)\, .
 \label{eq:rho-dispersion-transform-simplified}
\end{align}
\end{lemma}
\begin{proof}

Consider \eqref{eq:time-evolution-rho-dispersion}. If the frame of reference is translating the position by $\mathbf {r}_{\mathbf {k}_0}^t=\mathbf {r}_0+\mathbf{v_g}(\mathbf {k}_0)  t$ and the momentum by $\hbar \mathbf {k}_0$, we get the  simplified density functions \eqref{eq:rho-dispersion-transform-simplified}.

Theorem 2 in \cite{GeigerKedem2021b}
shows that  the  entropy in position and momentum is  invariant under translations of the position $\mathbf{r}$ and  the spatial frequency $\mathbf{k}$, and that completes the proof.
\end{proof}

The time invariance of the density $\rho^\entropyS_{\mathrm{k}}(\mathbf {k},t)$, and therefore of $\entropyS_{\mathrm{k}}$, reflects  the conservation law of momentum for free particles.

Coherent states, represented by state $\ket{\alpha}$,
are eigenstates of the annihilator operator. The 1D quantum phase space of observable variables $(x,p)$ can be constructed by the unitary  operator $U(x_0,p_0)=\eu^{\frac{\iu}{\hbar} (x_0 X -p_0 P) }$ applied to zero-state  $\ket{x=0,p=0}$, i.e., they can be constructed as  $\ket{\alpha}=\ket{x_0,p_0}=\eu^{\frac{\iu}{\hbar} (x_0 X -p_0 P)}\ket{0,0}$, where $\alpha=x_0+\iu p_0$. Projecting the state to  position space yields $\psi_{\alpha}(x)=\bra{x}\ket{\alpha}=\frac{\eu^{-\frac{p_0^2}{2}}}{\piu^{\frac{1}{4}}}\eu^{-\frac{1}{2}\left(x-\sqrt{2}\alpha\right)^2}$, where $\alpha=\frac{1}{\sqrt{2}}(x_0+\iu p_0)$. Squeeze states extend coherent states to all eigenstate solutions of the annihilator operator by allowing different variances to the Gaussian solution,  and together their representation in 3D position and momentum space are
\begin{align}
 \psi_{\mathrm{k}_0}(\mathbf {r}-\mathbf {r}_0)& = \bra{\mathbf {r}}\ket{\alpha}=\frac{1}{2^3 \piu^{\frac{3}{2}}(\det \matrixsym{\Sigmabold})^{\frac{1}{2}}}\, \normalx{\mathbf {r}}{\mathbf {r}_0}{\matrixsym{\Sigmabold} }\, \eu^{\iu \mathbf {k}_0\cdot \mathbf {r}}\,,
 \\
 \Phi_{\mathrm{r}_0} (\mathbf {k}-\mathbf {k}_0) &=\bra{\mathbf {k}}\ket{\alpha}=\frac{1}{2^3 \piu^{\frac{3}{2}}(\det \matrixsym{\Sigmabold}^{-1})^{\frac{1}{2}}}\, \normalx{\mathbf {k}}{\mathbf {k}_0}{\matrixsym{\Sigmabold}^{-1} }\, \eu^{\iu (\mathbf {k}-\mathbf {k}_0)\cdot \mathbf {r}_0}\,,
 \label{eq:coherent-state-3D}
\end{align}
where   $ \matrixsym{\Sigmabold} $ is the spatial covariance matrix.

\begin{theorem}
\label{thm:coherent-entropy}
A \QCurve with  an initial coherent state \eqref{eq:coherent-state-3D} and  evolving according to~ \eqref{eq:Dirac-Hamiltonian} is in \setIncreasing.
\end{theorem}
\begin{proof}
To describe the evolution of the initial states \eqref{eq:coherent-state-3D}, we apply~\eqref{eq:time-evolution-psi-dispersion}. Then, after
applying Lemma~\ref{lemma:simplified-densities},
\begin{align}
\rho^{\text{S}}_{\mathrm{r}}(\mathbf {r},t)& =\frac{1}{Z_2^2} \normalx{\mathbf {r}}{0}{\matrixsym{\Sigmabold} +\iu t \hessian(\mathbf {k}_0)} \normalx{\mathbf {r}}{0}{\matrixsym{\Sigmabold} -\iu t \hessian(\mathbf {k}_0)}
                                               = \normalx{\mathbf {r}}{0}{\frac{1}{2}\matrixsym{\Sigmabold}(t)}\,,
 \\
 \rho^{\text{S}}_{\mathrm{k}}(\mathbf {k},t) & = \normalx{\mathbf {k}}{0}{\matrixsym{\Sigmabold}^{-1} }\,,
\end{align}
where $\matrixsym{\Sigmabold}(t)=\matrixsym{\Sigmabold} + t^2 \hessian(\mathbf {k}_0) \matrixsym{\Sigmabold}^{-1}\hessian(\mathbf {k}_0)$.
Then
\begin{align}
 \entropyS & =\entropyS_{\mathrm{r}}+\entropyS_{\mathrm{k}}
 \\
           &= {-\kernBeforeIntegral\int \normalx{\mathbf {r}}{0}{\frac{1}{2}\matrixsym{\Sigmabold}(t)} \ln \normalx{\mathbf {r}}{0}{\frac{1}{2}\matrixsym{\Sigmabold}(t)}\diff^3\mathbf {r}}
  \\
           &
             \quad -\kernBeforeIntegral\int\normalx{\mathbf {k}}{0}{\matrixsym{\Sigmabold}^{-1}} \ln \normalx{\mathbf {k}}{0}{2\matrixsym{\Sigmabold}^{-1}}\diff^3\mathbf {k}
 \\
 &= 3 (1+  \ln \piu )  + \frac{1}{2}\ln \det\left ( \Identitymatrix+ t^2 (\matrixsym{\Sigmabold}^{-1} \hessian(\mathbf {k}_0))^2\right)\,  .
 \label{eq:entropy-coherent}
\end{align}
As $\det\left ( \Identitymatrix+ t^2 (\matrixsym{\Sigmabold}^{-1} \hessian(\mathbf {k}_0))^2\right)>0$,
the entropy increases over time.
\end{proof}
The theorem suggests that quantum physics has an inherent mechanism to increase entropy for free particles, due to the spatial dispersion property of the Hamiltonian.
Note that  at $t=0$ a coherent state \eqref{eq:coherent-state-3D} reaches the minimum possible  coordinate-entropy value, {while the spin-entropy  remains constant}.

\section{Time Reflection as a Mechanism to Convert \QCurves in \setIncreasing to \setDecreasing and Vice-Versa}
\label{subsec:time-reflection}

Consider a  time-independent Hamiltonian. We investigate the discrete symmetries C and P, and propose that Time Reversal  be augmented with Time Translation, say by $\deltau t$. We refer to  the  mapping $ t \mapsto -t + \deltau t$ as Time Reflection,  because as $t$ varies from $0$ to $\deltau t$,  $t'(t)=-t+\deltau t$ varies as a  reflection from $\deltau t$ to $0$. We  define the Time Reflection  quantum field
\begin{align}
\label{eq:time-reflection-T}
\Psi^{\mathrm{T}_{\deltau}}(\mathbf{r},-t+\deltau t)\eqdefA \mathscr{T} \Psi(\mathbf{r},t)= T \Psi^*(\mathbf{r},t)\,.
\end{align}

Note that  in contrast to the case of Time Reversal,   $\Psi^{\mathrm{T}_{\deltau}}(\mathbf{r},t')=\mathscr{T} \Psi(\mathbf{r},-t'+\deltau t)$, and the entropies associated with $\Psi(\mathbf{r},t)$ and  $\Psi^{\mathrm{T}_{\deltau}}(\mathbf{r},t)$  are generally not equal. Thus, an instantaneous Time Reflection transformation will cause entropy changes.

We next consider a composition of the  three transformation, Charge Conjugation, Parity Change, and Time Reflection.
\begin{definition}[$\Psi^{\CPT_{\deltau}}$]
  \label{def:cpt-deltau-quantum-field}
  Let the $\mathrm{CPT}_\deltau$ quantum field be
\begin{align}
\label{eq:CPT-delta-t}
    \Psi^{\mathrm{CPT}_{\deltau}}(-\mathbf{r},-t+{\deltau t})& \eqdefA \eta_{\deltau}\,  CPT\,  \overline{\Psi}^{\tran}(\mathbf{r},t)=\eta \gamma^5 \, (\Psi^{\dagger})^{\tran}(\mathbf{r}, t)\,,
\end{align}
where $\eta$ is the product of the phases of each operation,  $\eta_{\deltau}$ is the phase of time translation,
 and $\gamma^5\eqdefA\iu \gamma^0\gamma^1\gamma^2\gamma^3$.
\end{definition}

\begin{definition}[$Q_{\CPT_{\deltau}}$]
  \label{def:cpt-transformation}
  Let $Q_{\CPT_{\deltau}}$ be $\big(\psi(\mathbf{r},0), U(t), [0,\deltau t] \big)  \mapsto  \big(\psi^{\CPT_{\deltau}}(-\mathbf{r},0), U(t),  [0,\deltau t]  \big) $.
\end{definition}
Using \eqref{eq:CPT-delta-t} we see that,
\begin{align}
\psi^{\CPT_{\deltau}}(-\mathbf{r},0)\eqdefA \eta\gamma^5 \, (\Psi^{\dagger})^{\tran}(\mathbf{r}, -0+\deltau t)=\eta\, \gamma^5\, (\Psi^{\dagger})^{\tran}(\mathbf{r},\deltau t)\, .
    \label{eq:Qcpt}
\end{align}

\begin{theorem}[Time Reflection]
 \label{thm:CPT-QFT}
 Consider a $\mathrm{CPT}$ invariant quantum field theory (QFT) with energy conservation, such as Standard Model or Wightman axiomatic QFT \cite{wightman1976hilbert}. Let $e_0=\left (\psi(\mathbf{r},0), U(t), [0,\deltau t]\right)$ be a  \QCurve  solution to  such QFT.  Then, $e_{1} = Q_{\CPT_{\deltau}}(e_{0})$  is (i) a solution to such QFT, (ii) if $e_{0}$ is  in
 \setConstant,  \setDecreasing , \setOscillating, \setIncreasing then $e_{1}$ is respectively in \setConstant,  \setIncreasing , \setOscillating, \setDecreasing, making \setConstant,  \setIncreasing , \setOscillating, \setDecreasing reflections of \setConstant,  \setDecreasing , \setOscillating, \setIncreasing, respectively.
\end{theorem}

\begin{proof}
Let
$t'=-t+\deltau t$.
The \QCurve~$e_1$ describes  the evolution of  $\psi^{\CPT_{\deltau}}(-\mathbf{r},t')$
during  the period $[0,\deltau t]$.

{Since $e_0$ is a solution to a QFT that is $\text{CPT}$-invariant and time-translation invariant, $e_1$ is also a solution to the QFT, proving (i).}

The time evolution of $\psi^{\CPT_{\deltau}}(-\mathbf{r}, 0) $ from $0$ to $\deltau t$ is described by $\psi^{\CPT_{\deltau}}(-\mathbf{r}, t')$,  and  by~\eqref{eq:Qcpt}
$  \psi^{\CPT_{\deltau}}(-\mathbf{r}, t') = \eta\, \gamma^5\, (\Psi^{\dagger})^{\tran}(\mathbf{r},-t'+ \deltau t)=\eta\, \gamma^5\, \Psi^*(\mathbf{r}, \deltau t- t')$.  Thus, the evolution of $ \psi^{\CPT_{\deltau}}(-\mathbf{r}, t')$ as $t'$ evolves from $0$ to $\deltau t$,  by Theorem 3 in \cite{GeigerKedem2021b},  
has the same entropies as $\psi(\mathbf{r}, \deltau t -t')$.  Since $\psi(\mathbf{r}, \deltau t- t')$ traverses the same path  as $\psi(\mathbf{r}, t')$ but in the opposite time direction, we conclude that  $e_1$  produces the time evolution states  $\psi^{\CPT_{\deltau}}(-\mathbf{r}, t')$ in the time interval $[0,\deltau t]$ traversing the same path and with the same entropies as $\psi(\mathbf{r}, t')$, but in the opposite time directions.

Applying the above to  a \QCurve respectively in \setIncreasing, \setDecreasing,  \setConstant,  \setOscillating,  results in a \QCurve respectively in \setDecreasing, \setIncreasing,  \setConstant,  \setOscillating.
Thus, we conclude the proof of (ii).
 \end{proof}

\begin{figure}
 \centering
 \includegraphics[scale=0.4]{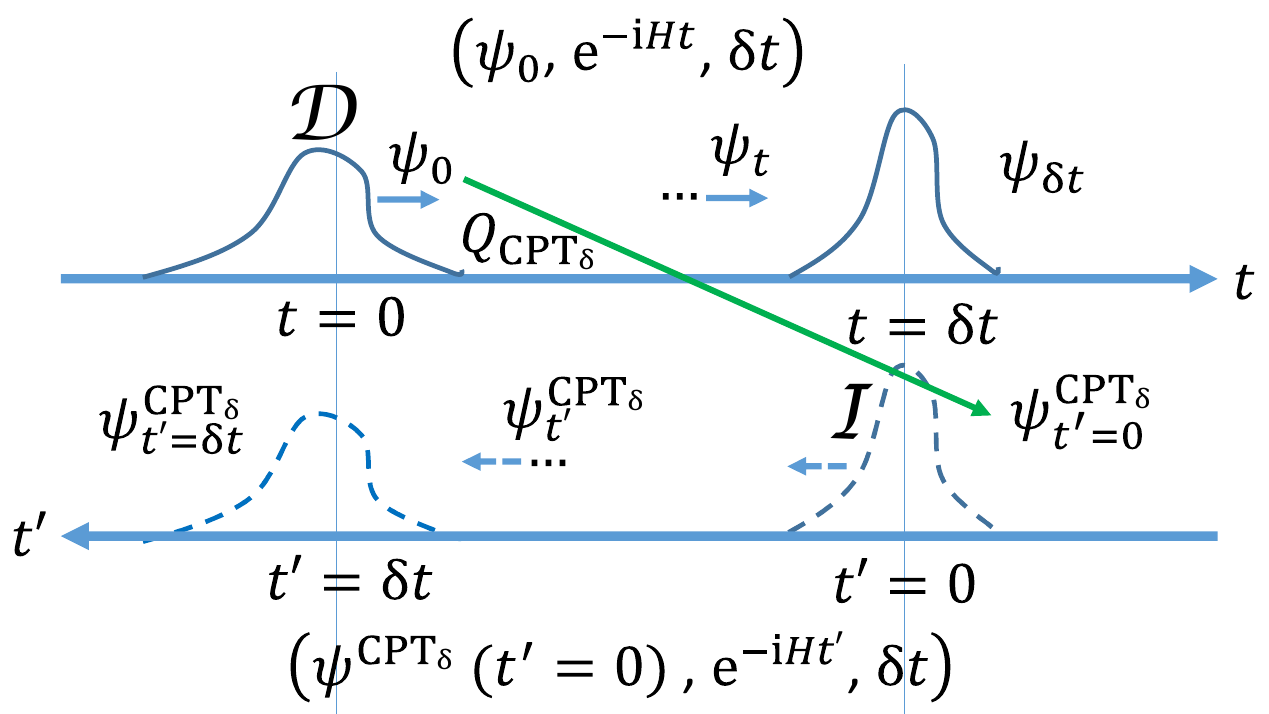}
 \caption{A visualization of the Time Reflection Theorem. (i) Axis $t$: A \QCurve $e_1=\big (\psi_0(\mathbf{r}), \eu^{-\iu H t}, \deltau t \big)$. (ii) Axis $t' =\deltau t-t$: The antiparticle \QCurve  is created as  $e_2=Q_{{\CPT}_{\deltau}}(e_1)=\big (
     \psi^{{\CPT}_{\deltau}}(-\mathbf{r},t'=0)
     , \eu^{-\iu H t'}, \deltau t\big)$. Axis $t'$ shows the evolution as going forward in time $t'$.  The evolution of $\psi^{{\CPT}_{\deltau}}(-\mathbf{r},t')=\eta \gamma^5 \, (\Psi^{\dagger})^{\tran}(\mathbf{r}, \deltau t - t')$ is mirroring the evolution of $\psi(\mathbf{r},t)$, with $t=t'$ evolving from $0$ to $\deltau t$. If  $e_1 \in \setDecreasing $, then  $e_2 \in \setIncreasing$.
}
 \label{fig:QCPT}
\end{figure}
For a visualization see Figure~\ref{fig:QCPT}.

\section{Entropy Oscillations}
\label{sec:oscillation}

Consider a Hamiltonian
  $H'  =H+H^{\mathrm{I}}$,
where   $|H^{\mathrm{I}}| \ll  |H|$,   and   the initial  eigenstate $\ket{\psi_{E_i}} $ of  $H$ associated with the eigenvalue $E_i=\hbar \omega_i$. The time evolution of $ \ket{\psi_{E_i}}$ is
\begin{equation}
     \ket{\psi_t}=\ee^{-\iu \frac{(H+H^{\mathrm{I}})}{\hbar} t} \ket{\psi_{E_i}}
                   =  \sum_{k=1}^n \alpha_k(t) \ket{\psi_{E_k}}\,,
\end{equation}
where $n$ is the number of the eigenvectors of $H$.  Fermi's golden rule \cite{fermi1950nuclear, dirac1927quantum} approximates the coefficients of transition for $k\ne i$ and short time intervals  by
\begin{align}
  \alpha_k(t)
\approx \frac{H^{\mathrm{I}}_{i,k} }{\hbar (\omega_{i}-\omega_{k})} \left(-2\sin^2 \left (  \frac{(\omega_i-\omega_k) t}{2}\right)+\iu  \sin \left (  (\omega_i-\omega_k) t\right)  \right)\,.
\end{align}

\begin{theorem}[Entropy Oscillations]
  \label{theorem:oscillations}
  Consider the \QCurve  $ \big(\ket{\psi_{E_i}}, U(t)=\ee^{-\iu \frac{(H+H^{\mathrm{I}})}{\hbar} t}, T \big)$  with  $\hbar \omega_{1}$  the ground state value of   $H$ and $T= \frac{2\piu }{|\omega_{i}-\omega_{1}|}$. Assume that $|\alpha_1(t)|^2, |\alpha_i(t)|^2 \gg |\alpha_k(t)|^2\, \text{for} \, k\ne 1,i$ and $t\in [0, T]$. Then the \QCurve   is in $\setOscillating$.
\end{theorem}
\begin{proof}
  With the theorem's assumptions, we can approximate the position and the momentum probability densities  associated with $\ket{\psi_t}$ by
\begin{align}
     \rho_{\mathrm{r}}(\mathbf{r},t)  &\approx \left |\sqrt{1-|\alpha_1(t)|^2} \bra{\mathbf{r}}\ket{\psi_{E_i}}+\alpha_1(t) \bra{\mathbf{r}}\ket{\psi_{E_1}}\right |^2\,,
     \\
     \rho_{\mathrm{k}}(\mathbf{k},t) &\approx \left |\sqrt{1-|\alpha_1(t)|^2} \bra{\mathbf{k}}\ket{\psi_{E_i}}+\alpha_1(t) \bra{\mathbf{k}}\ket{\psi_{E_1}}\right |^2\,.
   \end{align}
The time coefficients  are the same for $\rho_{\mathrm{r}}(\mathbf{r},t)$ and $\rho_{\mathrm{k}}(\mathbf{k},t)$, and they all return to the same values simultaneously  after a period of $T$, and so the entropy will return to its previous value  too. As the entropy is not a constant, it must be oscillating.
\end{proof}

Thus, when Fermi's golden rule can be applied, the coefficients of the transition probabilities of the unitary evolution of a state oscillate, and the entropy associated with the evolution of such a state will also oscillate with the same period.

\begin{theorem}[Coefficients for two states]
\label{theorem:FermiCoefficientsExactTwoState}
Consider a  particle in  an eigenstate  $\ket{\psi_{E_1}}$ of  a Hamiltonian $H$ that has only two eigenstates $\ket{\psi_{E_1}}$ and $ \ket{\psi_{E_2}}$  with eigenvalues $E_1=\hbar \omega_1$ and $E_2=\hbar \omega_2$, respectively. Let this particle interact with an external field (such as the impact of a  Gauge Field), requiring an additional Hamiltonian term $H^{\mathrm{I}}$ to describe the evolution of this system.

Let
$\omega^{\mathrm{I}}_{i,j} = \frac{1}{\hbar}\bra{\psi_{E_i}}H^{\mathrm{I}}\ket{\psi_{E_j}}$,
$\omega_{1}^{\mathrm{total}} = \omega_1+\omega_{11}^{\mathrm{I}}$,
$\omega_{2}^{\mathrm{total}} = \omega_2+\omega_{22}^{\mathrm{I}}$,
$\eta = \sqrt{\left(\omega_{1}^{\mathrm{total}} - \omega_{2}^{\mathrm{total}} \right)^2 +4 (\omega_{12}^{\mathrm{I}})^2 }$,
and
$\lambda_{\pm} = \frac{\omega_{1}^{\mathrm{total}} + \omega_{2}^{\mathrm{total}} \pm \eta}{2}$.
  The probability of the particle to be in  state $\ket{\psi_{E_2}}$ at time $t$ is
  \begin{equation}
    \label{eq:1}
\frac{4(\omega_{12}^{\mathrm{I}})^2  }{\eta^2} \, \sin^2 \frac{(\lambda_{+}-\lambda_{-}) t}{2}\,.
  \end{equation}
\end{theorem}

\begin{proof}
The Hamiltonians in the basis $\ket{\psi_{E_1}}, \ket{\psi_{E_2}}$ are
\begin{align}
    H = \hbar \begin{pmatrix} \omega_1 & 0 \\ 0 & \omega_2 \end{pmatrix}\qquad \text{and} \qquad  H^{\mathrm{I}} = \hbar \begin{pmatrix} \omega_{11}^{\mathrm{I}} & \omega_{12}^{\mathrm{I}} \\ \omega_{12}^{\mathrm{I}} & \omega_{22}^{\mathrm{I}} \end{pmatrix}\,,
    \label{eq:Hamiltonians}
\end{align}
where the real values satisfy  $\omega_{21}^{\mathrm{I}}=\omega_{12}^{\mathrm{I}}$ as $H^{\mathrm{I}}$ is Hermitian. The  eigenvalues of the symmetric matrix $H' = H+H^{\mathrm{I}}$ are  $\hbar \lambda_{\pm}$, and so we can decompose it as
\begin{align}
   H' &= \hbar \begin{pmatrix}  \omega_{1}^{\mathrm{total}}& \omega_{12}^{\mathrm{I}} \\ \omega_{12}^{\mathrm{I}} & \omega_{2}^{\mathrm{total}} \end{pmatrix}
   =\begin{pmatrix}
        \cos \theta &  -\sin \theta
        \\
         \sin \theta &  \cos \theta
    \end{pmatrix} \begin{pmatrix}
      \hbar \lambda_{+} &  0
        \\
         0 &  \hbar\lambda_{-}
    \end{pmatrix} \begin{pmatrix}
        \cos \theta &  \sin \theta
        \\
         -\sin \theta &  \cos \theta
    \end{pmatrix}\, ,
    \label{eq:H01-decompose}
\end{align}
where
\begin{equation}
  \label{eq:theta-definition}
  \theta =  \frac{1}{2} \arcsin \frac{2\omega_{12}^{\mathrm{I}}}{\eta} \,.
\end{equation}

The time evolution of $ \ket{\psi_{E_1}}$ is
$     \ket{\psi_t}=\ee^{-\iu \frac{(H+H^{\mathrm{I}})}{\hbar} t} \ket{\psi_{E_1}}
                   =  \sum_{k=1}^2 \alpha_k(t) \ket{\psi_{E_k}}$,
and projecting on
    $\bra{\psi_{E_j}}$, we get
$    \alpha_j(t)  =  \bra{\psi_{E_j}}\ee^{-\iu \frac{(H+H^{\mathrm{I}})}{\hbar} t} \ket{\psi_{E_1}}$.
From \eqref{eq:H01-decompose},
\begin{align}
   \ee^{-\iu \frac{H'}{\hbar} t} &=\begin{pmatrix}
        \cos \theta &  -\sin \theta
        \\
         \sin \theta &  \cos \theta
    \end{pmatrix} \begin{pmatrix}
     \ee^{-\iu  \lambda_+ t} &  0
        \\
         0 &  \ee^{-\iu  \lambda_- t}
    \end{pmatrix} \begin{pmatrix}
        \cos \theta &  \sin \theta
        \\
         -\sin \theta &  \cos \theta
    \end{pmatrix}
    \\
    &=\begin{pmatrix}
         \ee^{-\iu  \lambda_+ t}\cos^2 \theta  + \ee^{-\iu  \lambda_- t}\sin^2 \theta & \,  \frac{\ee^{-\iu  \lambda_+ t} - \ee^{-\iu  \lambda_- t}}{2} \sin 2 \theta
        \\
         \frac{\ee^{-\iu  \lambda_+ t} - \ee^{-\iu  \lambda_- t}}{2} \sin 2 \theta & \,   \ee^{-\iu  \lambda_+ t}\sin^2 \theta  + \ee^{-\iu  \lambda_- t}\cos^2 \theta
    \end{pmatrix}\,.
\end{align}
Thus,
\begin{align}
\label{eq:Fermi-coefficient-two-state}
     \begin{pmatrix}
    \alpha_1(t)
    \\
    \alpha_2(t)
    \end{pmatrix} &=  \ee^{-\iu \frac{H'}{\hbar} t}   \begin{pmatrix}
    1
    \\
    0
    \end{pmatrix}=\begin{pmatrix}
         \cos^2 \theta \, \ee^{-\iu  \lambda_+ t}  + \sin^2 \theta \,  \ee^{-\iu  \lambda_- t}
        \\
         \sin 2 \theta \,  \left(\frac{\ee^{-\iu  \lambda_+ t} - \ee^{-\iu  \lambda_- t}}{2} \right)
    \end{pmatrix}\, ,
\end{align}
and so
\begin{align}
    \begin{pmatrix}
    |\alpha_1(t)|^2
    \\
    |\alpha_2(t)|^2
    \end{pmatrix} &= \begin{pmatrix}
     1-\frac{1}{2} \sin^2 2\theta \left (1-  \cos (\lambda_--\lambda_+) t\right)
    \\
  \frac{1}{2}  \sin^2 2 \theta \, \left (1-  \cos (\lambda_--\lambda_+) t\right)
    \end{pmatrix}\,.
\end{align}
As $1-  \cos (\lambda_--\lambda_+) t =2\, \sin^2 \frac{(\lambda_+-\lambda_-) t}{2}$ ,  the probability of being in  state $\ket{\psi_{E_2}}$ at time $t$  is $|\alpha_2(t)|^2=\sin^2 2 \theta \, \sin^2 \frac{(\lambda_+-\lambda_-) t}{2}$. Using \eqref{eq:theta-definition}, completes the proof.
\end{proof}

If $\omega_{1} \gg \omega_{11}^{\mathrm{I}}$, $\omega_{2} \gg \omega_{22}^{\mathrm{I}}$, and $|\omega_{1}- \omega_{2}| \gg \omega_{12}^{\mathrm{I}}$, then $\lambda_{+,-} \approx \omega_{1,2}$, and the coefficient of transition becomes $|\alpha_2(t)|^2\approx \frac{4(\omega_{12}^{\mathrm{I}})^2  }{(\omega_{1}-\omega_{2})^2} \, \sin^2 \frac{(\omega_2-\omega_1) t}{2}$, which is  Fermi's golden rule \cite{fermi1950nuclear, dirac1927quantum}.

\section{A Two-Particle Collision}
\label{subsec:two-particle-system}

Consider a  two-fermions or a two-massive-bosons system
 \begin{align}
  \label{eq:two-particle-state}
   \ket{\psi_t}&=\frac{1}{\sqrt{C_t}}  \left (\ket{\psi^1_t}\ket{\psi^2_t}\mp \ket{\psi^2_t}\ket{\psi^1_t}
  \right )\,,
  \end{align}
where $C_t$ is the normalization constant that may evolve over time and the signs ``$\mp$'' represent fermions (``$-$'') and bosons (``+''). When  $\ket{\psi^1_t}$ and $ \ket{\psi^2_t}$  are orthogonal to each other,  $C_t=2$.
 Projecting on $\bra{\mathbf {r}_1}\bra{\mathbf {r}_2}$ and on $\bra{\mathbf {k}_1}\bra{\mathbf {k}_2}$,
 \begin{align}
     \psi(\mathbf {r}_1,\mathbf {r}_2 ,t)&=\frac{1}{\sqrt{C_t}}
     \left ( \psi_1(\mathbf {r}_1,t) \psi_2(\mathbf {r}_2,t)\mp \psi_1(\mathbf {r}_2,t) \psi_2(\mathbf {r}_1,t)  \right )\,,
      \\
     \psi(\mathbf {k}_1,\mathbf {k}_2 ,t)&=\frac{1}{\sqrt{C_t}}
     \left ( \phi_1(\mathbf {k}_1,t) \phi_2(\mathbf {k}_2,t)\mp \phi_1(\mathbf {k}_2,t) \phi_2(\mathbf {k}_1,t)  \right )\,.
 \end{align}

The  entropy of the two-particle system, discarding the spin-entropy which is constant throughout the collision,   is then
  \begin{align}
     \label{eq:entropy-two-particles}
     S\left (\ket{\psi^1_t},\ket{\psi^2_t}\right )      &= - \int \! \diff^3 \mathbf {r}_1 \int \diff^3 \mathbf {r}_2\,   \rho_{\mathrm{r}}(\mathbf {r}_1,\mathbf {r}_2 ,t)
     \ln \rho_{\mathrm{r}}(\mathbf {r}_1,\mathbf {r}_2 ,t)
     \\
     & \quad - \int \! \diff^3 \mathbf {k}_1 \int \diff^3 \mathbf {k}_2\,   \rho_{\mathrm{k}}(\mathbf {k}_1,\mathbf {k}_2 ,t)
     \ln \rho_{\mathrm{k}}(\mathbf {k}_1,\mathbf {k}_2 ,t)\, .
  \end{align}
Consider a collision of two particles,  each one described by  an initial coherent state  with position variance $\sigma^2$ centered at $c_1$ and $c_2$  and  moving towards each other along the $x$-axis  with center momenta $ {p}_0=\hbar {k}_0$ and $ - {p}_0$. They can be represented in position and momentum space as
\begin{align}
\Psi_{1}(x,t) &=\frac{\eu^{-\iu k_0 v_p(k_0) \, t}}{Z_1}\,  \normalx{x}{c_1 +v_g(k_0) \, t}{\sigma^2  +\iu \, t\, {\cal H}(k_0) }\eu^{\iu k_0 x} \,,
 \\
\Psi_{2}(x,t) &=\frac{\eu^{-\iu k_0 v_p(k_0) \, t}}{Z_1}\,\,  \normalx{x}{c_2 -v_g(k_0) \, t}{\sigma^2 +\iu \, t\, {\cal H}(-k_0) } \eu^{-\iu k_0 x}\, ,
\\
 \Phi_1(k, t) & =   \frac{\eu^{-\iu t\, v_p(k_0) k_0} }{Z_{k_0}} \, \, \normalx{k}{k_0}{(\sigma^2  +\iu \, t\, {\cal H}(k_0))^{-1} }\, \eu^{\iu (k-k_0)\left (c_1+v_g(k_0) \, t\right)}\,,
\\
 \Phi_2(k, t) & =   \frac{\eu^{-\iu t\, v_p(k_0) k_0} }{Z_{k_0}} \, \, \normalx{k}{-k_0}{(\sigma^2  +\iu \, t\, {\cal H}(-k_0))^{-1} }\, \eu^{\iu (k+k_0) \left (c_2-v_g(k_0) \, t\right)}\, .
\label{eq:two-coherent-states}
\end{align}
\begin{figure}
 \centering
\includegraphics[scale=0.4]{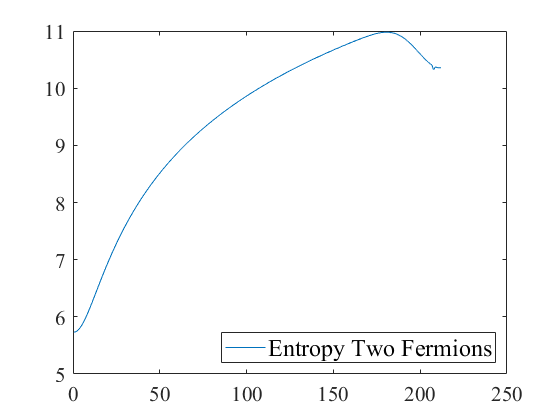}\includegraphics[scale=0.4]{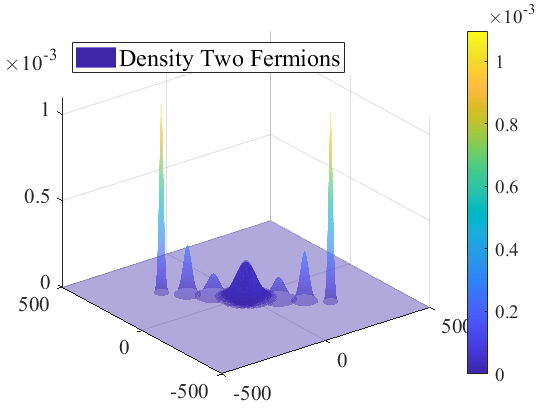}
 \\
 \quad {\small (a) $\frac{\hbar}{m}=1$ : Entropy vs. time;   \;   $\rho_{x}(x_1,x_2 ,t)$  overlaid over time}
 \\
 \includegraphics[scale=0.4]{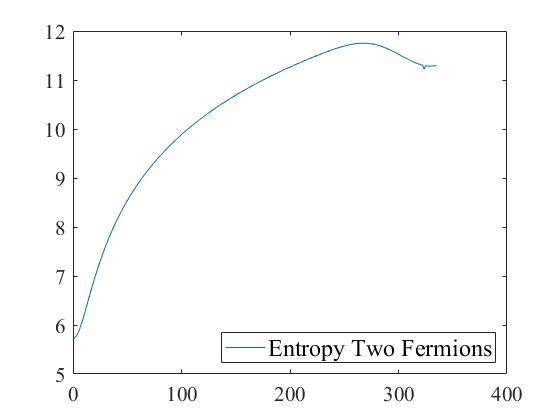}\includegraphics[scale=0.4]{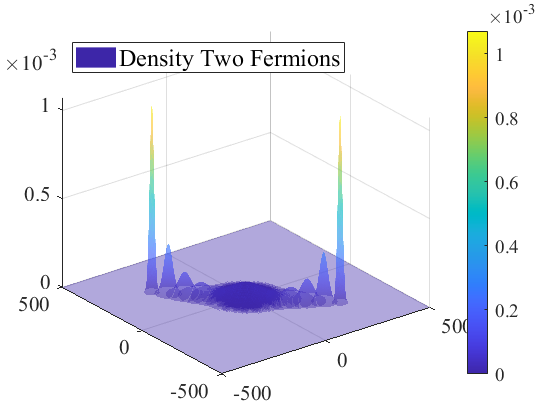}
 \\
 \quad {\small (b) $\frac{\hbar}{m}=0.5$ : Entropy vs. time;  \;  $\rho_{x}(x_1,x_2 ,t)$  overlaid over time}
 \\
 \caption{Collision of two fermions with individual amplitudes \eqref{eq:two-coherent-states}, parameters $k_0=1$, $c_2=-c_1=300$, speed of light $c=1$, a grid of $1\,000$ points for $x_1, x_2,k_1,k_2$. The left column shows entropy vs. time. The right column shows snapshots of the density at initial time, final time, and intervals of 100 time units,  overlaid on single plots. The $z$-axis represents the density, and the $x$-$y$ axes represent the $x_1$-$x_2$ values. As the particles approach each other, their individual densities disperse, the maximum values are reduced, and the entropy increases. Only when the particles  are close to each other, the interference reduces the total entropy.}
 \label{fig:two-particle-collision-entropy}
\end{figure}

Figure~\ref{fig:two-particle-collision-entropy} shows that when  the two particles  are far apart, the  entropy of the system is close to the sum of the two individual  entropies, with each one increasing over time.
The spatial entanglement decreases the uncertainty, and therefore the entropy too.
The competition between the increase of the entropy of the individual particles and the decrease of the entropy due to entanglement results in an oscillation and the decrease in the total entropy when the two particles are close to each other.

\section{An Entropy Law and a Time Arrow}

In classical statistical \mechanics, the entropy provides a time arrow through the second law of thermodynamics \cite{clausius1867mechanical}.  We have  shown that due to the dispersion property of the fermions Hamiltonian some states in quantum \mechanics, such as  coherent states,  already obey such a law.   However, current quantum physics is described as time reversible. In \cite{GeigerKedem2021b} we conjectured the following

\begin{law*}[The Entropy Law]
 \label{postulate:1}
 The entropy of an isolated quantum system is an increasing function of time.
\end{law*}
It is an information-theoretic conjecture about isolated quantum states, whereby information (the inverse of the entropy) cannot be gained.  This law implies a time arrow for quantum states, noting that states where the time evolution of the entropy is constant may be reversible. 
 
Moreover, this law would imply that  a quantum state whose  evolution is governed by an entropy that oscillates, would be blocked from entering a time interval of decreasing entropy.  Instead, a   collapse to a new state with increasing entropy would occur, where the state is one of the superposition states selected according to its  probability. This is the case for the scenario studied in section dedicated to the experiment with a cavity and an atom.
The law may help explain why particles are created and/or annihilated in scenarios such as
high-speed collision $\ee^+ + \ee^- \rightarrow 2 \upgamma$, kaons decay into mesons, and photon creation and emission when  the electron in the hydrogen atom transitions from an excited state  to the ground state. In those scenarios, while
such final states are reachable in a unitary evolution of the initial state, it seems that only those evolutions in which entropy increase are realized.
According to the S-matrix formulation \cite{weinberg1995quantum1}, similar to Fermi's golden rule in QM, these final states are among the  superposition states and are selected according to the probability of the superposition.  These scenarios suggest that the  creation and/or annihilation  of a particles occur when
the entropy of the evolution from the initial to the final state is oscillating and it occurs to interrupt the oscillation before the entropy would decrease.

\section{Conclusions}
\label{sec:conclusion}

The concept of entropy in quantum phase spaces proposed in ~\cite{GeigerKedem2021b}  were further developed here. 
We analyzed the entropy evolution  through the partition of \QCurves into the  four sets   \setConstant,  \setIncreasing , \setOscillating, \setDecreasing. We showed that   the  Dirac's Hamiltonian  disperses  information due to its positive Hessian, causing  coherent states time evolution to  increase entropy. We proved that Time Reflection transforms \QCurves in \setConstant,  \setIncreasing , \setOscillating, \setDecreasing into  \QCurves in \setConstant,  \setDecreasing , \setOscillating, \setIncreasing, respectively.  We proved that  an initial eigenstate of a Hamiltonian evolving with  the addition of a Hamiltonian term not only causes  a state oscillation (as suggested by Fermi's golden rule when the appropriate approximations hold) but also  causes entropy oscillation.  We studied collisions of two particles, each evolving as a coherent state, and  showed that as they come closer to each other the total system's entropy oscillates. The results are applicable to both the Quantum Mechanics (QM) and the Quantum Field Theory (QFT) settings, but we  generally presented them in the more convenient setting. 

We reviewed the conjectured entropy law ~\cite{GeigerKedem2021b} which assigns a time arrow.  We wonder if the  entropy for all quantum states under a unitary evolution is  increasing (such as the case for coherent states under free Dirac Hamiltonian evolution), and the time reverse process would be discarded by this law. Or, whether there are quantum states that under a unitary evolution have the entropy oscillating. In this case, this law would block such evolution when the entropy would decrease triggering the collapse of the state to a new state where the entropy increases under unitary evolution.  We analysed experiments with atoms in a cavity where an apparent oscillation occurs. We argued that instead, that the high entropy photon emitted is absorbed by the cavity with a new creation of a low entropy photon that is fed back to the atom and we speculated that in that scenario the entropy is  always increasing.

\section{Acknowledgement} This paper is partially based upon work supported by both the National Science Foundation under Grant No.~DMS-1439786 and the Simons Foundation Institute Grant Award ID 507536 while the first author was in residence at the Institute for Computational and Experimental Research in Mathematics in Providence, RI, during the spring 2019 semester ``Computer Vision''  program.

 \bibliographystyle{abbrv}
 \bibliography{gk01}

\end{document}